\documentclass[10pt,letterpaper,twocolumn,twoside,dvipdfm,conference]{IEEEtran}

\makeatletter
\def\ps@headings{%
\def\@oddhead{\mbox{}\scriptsize\rightmark \hfil \thepage}%
\def\@evenhead{\scriptsize\thepage \hfil \leftmark\mbox{}}%
\def\@oddfoot{}%
\def\@evenfoot{}}
\makeatother
\usepackage{cite}
\usepackage{times}
\usepackage{dsfont}
\usepackage{cite}
\usepackage{times}

\usepackage{graphicx}
\usepackage{graphics}
\usepackage{subfigure}
\usepackage{epsfig}
\usepackage{latexsym}
\usepackage{amsfonts}
\usepackage{amssymb}
\usepackage{paralist}
\usepackage{comment}
\usepackage{xspace}
\usepackage{mathrsfs}
\usepackage{amssymb}
\usepackage{slashbox}
\usepackage{color}
\usepackage[ruled]{algorithm}
\usepackage[noend]{algorithmic}

\usepackage{amssymb}
\usepackage{url,epsfig,array}
\usepackage{leftidx}
\usepackage{amsmath}

\def\ie{\textit{i.e.}\xspace}

\def\etc{\textit{etc.}\xspace}

\def\eg{\textit{e.g.}\xspace}

\newtheorem{theorem}{Theorem}

\newtheorem{lemma}[theorem]{Lemma}

\begin{document}
\bibliographystyle{plain}

\title{Truthful Auction Mechanism for Heterogeneous Spectrum Allocation in Wireless Networks}

\author{\IEEEauthorblockN{He Huang\IEEEauthorrefmark{1}, Yu-e Sun\IEEEauthorrefmark{2}, Xiang-Yang Li\IEEEauthorrefmark{3}, Hongli Xu\IEEEauthorrefmark{4}, Yousong Zhou\IEEEauthorrefmark{4} and Liusheng Huang\IEEEauthorrefmark{4}}
\IEEEauthorblockA{\IEEEauthorrefmark{1}School of Computer Science and Technology, Soochow University, Suzhou, China\\
\IEEEauthorrefmark{2}School of Urban Rail Transportation, Soochow University, Suzhou, China\\
\IEEEauthorrefmark{3}Department of Computer Science, Illinois
Institute of Technology, and TNLIST, Tsinghua University.\\
\IEEEauthorrefmark{4}Department of Computer Science and Technology, University of Science and Technology of China, Hefei, China\\
Emails: \{huangh, sunye12\}@suda.edu.cn, xli@cs.iit.edu, lshuang@ustc.edu.cn}}

\maketitle

\begin{abstract}
Secondary spectrum auction is widely applied in wireless networks for
mitigating the spectrum scarcity. In a realistic spectrum trading
market, the requests from secondary users often specify the usage of a
fixed spectrum frequency band in a certain geographical region and
require a duration time in a fixed available time
interval. Considering the selfish behaviors of secondary users, it is
imperative to design a truthful auction which matches the available spectrums
and requests of secondary users optimally. Unfortunately, existing
designs either do not consider spectrum heterogeneity or ignore the
differences of required time among secondary users.

In this paper, we address this problem by investigating how to use
auction mechanisms to allocate and price spectrum resources so that
the social efficiency can be maximized. We begin by classifying the
spectrums and requests from secondary users into different local
markets which ensures there is no interference between local markets,
and then we can focus on the auction in a single local market. We
first design an optimal auction based on the Vickrey-Clarke-Groves (VCG)
mechanism to maximize the social efficiency while enforcing
truthfulness. To reduce the computational complexity, we further
propose a truthful sub-optimal auction with polynomial time complexity,
which yields an approximation factor $6+4 \sqrt{2}$.
Our extensive simulation results using real spectrum availability data
show that the social efficiency ratio of the sub-optimal auction is
always above 70\% compared with the optimal auction.

\end{abstract}

\IEEEpeerreviewmaketitle

\section{Introduction}\label{sec1}
With the increasing popularity of wireless devices
and applications, the ever-increasing demand of traffic poses a great
challenge in spectrum allocation and usage. However, current fixed
\emph{long-term} and \emph{regional lease} spectrum allocation scheme
leads to significant spectrum \emph{white spaces} and artificial
shortage of spectrum resources. Many efforts such as the Federal
Communications Commission (FCC) ruling on white spaces are attempting
to free the under-utilized licensed spectrum by permitting
opportunistic access \cite{jones2008evaluation}. However, the
incumbents have no incentive to permit their spectrum to be shared
\cite{kash2011enabling}. In order to utilize such idle spectrum
resources, one promising technology is to encourage secondary users
sublease spectrum from primary users (who own the right to use
spectrum exclusively) \cite{xu2010salsa}.

Auction serves as such an effective way that helps increase the
efficiency and effectiveness of the spectrum, in which the primary
user could gain utilities by leasing their idle spectrums in economic
perspective while new applicants could gain access to these spectrums
\cite{ji2008multi},\cite{kasbekar2010spectrum}. Previous studies on
spectrum auctions (\eg \cite{huangsprite},\cite{zhou2008ebay},\cite{zhou2009trust}) mainly consider wireless interference and spatial
reuse of channels under economic robustness constraints. Most of the
existing works assume that secondary users can share one channel only
if they are spatial-conflict free with each other. However, under a
more realistic model, a secondary user may only be interested in
the usage of one channel during some specific time periods. Therefore,
secondary users can share the same channel in spatial, temporal, and
spectral domain without causing interference with each other.
So, it is reasonable to further
improve the spectrum utilization by introducing time-domain.

Following in this direction, some papers \cite{wang2010toda},\cite{xu2010salsa},\cite{xu2011tofu} and \cite{chen2009mining} take the requested time durations of secondary users into consideration. However, all these works only consider a special case where all the secondary users request some \emph{fixed} continuous time intervals. In fact, the request time of a secondary user is not always fixed. For example, some people may request the usage of one channel for 2 hours in an available time interval which last from 2:00PM to 5:00PM, instead of requesting a fixed time interval lasts from 2:00PM to 4:00PM. Therefore, the case of a secondary user requests a duration time for the usage of one channel in some available time windows is more general than the study of previous works. On the other side, spectrum provided by the primary users is often \emph{heterogeneous} in a realistic spectrum market. For instance, spectrums may reside in various frequency bands, and the communication quality changes greatly when frequency band varies (\emph{frequency heterogeneity})\cite{fengtahes}. Meanwhile, spectrum is a local
resource and is available only within the license region (\emph{market locality}) \cite{wangdistrict}.
Spectrum heterogeneity is investigated in
\cite{fengtahes},\cite{kash2011enabling},\cite{parzy2011non}. Nevertheless, none of the existing works has addressed the
spectrum heterogeneity and secondary users' time demand at the same time.



In this paper, we propose a framework in which  secondary users can
request the usage of one channel with specific \emph{frequency band
  type} in a specific area and during some specific time periods. The
time slots allocated to a fixed request should  be supplied by one
channel, and may be discretely distributed in the specific time
interval that secondary user asked. A natural goal of spectrum auction is to maximize the social efficiency, \ie, allocating spectrum to the secondary users who value it most.
Therefore, our aim in this work is to design auction
mechanisms which maximize \emph{social efficiency} while ensuring truthful
bidding from secondary users. This model is similar to the weighted
time scheduling problem for the multi-machine version
\cite{bar2002approximating}. However, the studies on time scheduling
problems are not concerned with the truthfulness of jobs, which is one of
the most critical properties in spectrum auction.
Furthermore, in \cite{bar2002approximating} a fixed job can
be allocated into different machines in time scheduling problem. In
our model, a fixed request can only be allocated in one
channel. Therefore, the studies on time schedule problem cannot be
directly used in auctions for spectrum allocation. To the best of our
knowledge, we are the first to design truthful auction mechanisms for
spectrum allocation in this model.

The main contributions of this paper are as follows. This paper
studies the case where each request of secondary user contains a
specific area and an interested type of spectrum. Moreover, channels
supplied by the primary user also include sub-region and spectrum type
information. We divide the spectrum market into several
non-interference local markets according to the area and spectrum
type. Our new schemes focus on the auction in a single local
market. Assuming that the \emph{conflicting model} of secondary users
in a specific channel is a complete graph, we first design an optimal
auction based on the Vickrey-Clarke-Groves (VCG) mechanism to enforce
truthfulness and maximization of the social efficiency. As everyone knows, the
winner determining problem with VCG mechanism is a NP-hard
problem. Therefore, we further propose a sub-optimal auction with a
\emph{greedy-like} winner determination mechanism and a \emph{critical
  value} based payment rule, which together induce truthful
bidding. We will show that our sub-optimal auction has a polynomial time
complexity and yields a constant approximation factor at most $6+4 \sqrt{2}$.
The low time complexity makes this auction much more practical for large scale
spectrum market.
We then conduct extensive simulation studies on the performance
 of our mechanisms using real spectrum availability data.
Our simulation results show that the performance of our
sub-optimal mechanism is efficient in social efficiency compared with
the optimal VCG method.
The social efficiency achieved by our suboptimal method is actually more than $70\%$
 of the optimal.


The rest of paper is organized as follows: Section \ref{sec2}
introduces preliminaries and our design targets. Sections \ref{sec3}
and \ref{sec4} propose our algorithm design for optimal and suboptimal mechanisms.
Section \ref{sec5} evaluates the performance of our mechanisms. Section \ref{sec6} reviews
related work and Section \ref{sec7} concludes the paper.

\section{Preliminaries}\label{sec2}

\subsection{Heterogeneous Spectrum Auction Model}\label{sec2A}

Consider a spectrum setting where one auctioneer (Primary User) contributes $m$ distinct channels to $n$ secondary users located in a geographic region $\mathcal{L}$. The auctioneer holds the usage right of $m$ spectrums $\mathcal{S}=\{s_1 , s_2 , ...,s_m \}$  and is willing to sublease the usage of these channels to secondary users for time intervals. The auction system consists of $n$ secondary users $\mathcal{B}=\{b_1 , b_2 ,..., b_n \} $ who want to pursue the right of using some channels for some period of time. Based on the inherent characteristic of spectrum market locality (\emph{spatial heterogeneity}), we make the assumption that the entire licensed region $\mathcal{L}$  can be partitioned into several disjoint sub-regions. Each disjoint geographical sub-region $i$ is denoted by $l_i$. Fig. \ref{fig:partition} shows a sub-region partition instance. In our model, each channel is only available in the sub-regions which it can be used at the same time. Therefore, auctions happened in different sub-regions do not influence each other.

\begin{figure}[!t]
\centering
\includegraphics[width=2in,height=1.2in]{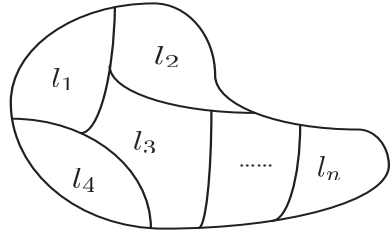}
\caption{The  licensed region $\mathcal{L}$ is partitioned into several disjoint sub-regions.}
\label{fig:partition}
\end{figure}

As we know, spectrum request from secondary user often specify a particular frequency band they needed in a practical spectrum market. For example, some wireless users only request the spectrums residing in lower-frequency bands due to the limitation of wireless devices. Therefore, we will also take \emph{frequency heterogeneity} into consideration in our auction model. Since that the candidate spectrums can be grouped into different spectrum type sets $\mathcal{ST}$ based on their frequency bands, we can partition the whole secondary spectrum market into several local markets.
All the spectrums in local market $M_k$ can be used in the same sub-region and with the same spectrum type. Since auctions in different local markets have no mutual effects with each other, we can just focus on the auction in a single local market $M_k$.

\subsection{Spectrum Bidding Model}\label{sec2B}

Assume a secondary user $b_i \in \mathcal{B}$  has a set of specific spectrum requests $\mathcal{J}_i$ . Each spectrum request $I_j \in \mathcal{J}_i$ can be defined as a \textbf{job}. A secondary user $b_i$ can bid for several distinct jobs in multiple local markets, but only one job at most in a specific local market. Each $I_j \in \mathcal{J}_i $  can be described as  $I_j = (l_j , ST_j , v_j , a_{j} , d_{j}, t_j)$, where $l_j$ shows the preferred spectrum release region, $ST_j$ explains the interested spectrum type, and $v_j$ denotes the bidding price for the usage right of specific channel.
$a_j$ , $d_j$ and $t_j$ respectively denote each job's arrival time, deadline and duration (or job length). Note that the allocation time slots for each job can only stem from a single spectrum, and the request time interval is not necessarily continuous.

Each candidate channel  $s_i \in \mathcal{S} $ provided by the primary user can be characterized by a triple $(l_i , ST_i , A_i) $, where  $l_i$ denotes the sub-region where $s_i$  accommodates,  $ST_i$ describes the spectrum type of $s_i$ , and $A_i$ includes all the available time slots in $s_i$. Primary users will set $\eta_s$ as the per-unit reservation price of each spectrum. 


\subsection{Economic Requirements and Design Target}\label{sec2C}

In this paper, we will study the complete conflict graph model for secondary users in each $s_i$, and leave the general conflict graph model as a future work. The objective of this work is to design a heterogeneous spectrum auction mechanism satisfying the necessary economic property of truthfulness (a.k.a \emph{Strategyproofness}), which maximizes the social efficiency at the same time.

We use $\mathbf{f}$ to denote the vector of all bids in an auction, and use $\mathbf{f}_{-j}$ to denote the set of bids for all jobs except job $I_j$. Each job $j$ is charged a payment $p_j$ if $I_j$ wins in the auction. Thus, the utility for job $I_j$ can be stated as:

\begin{equation*}
u_j=\left\{ \begin{array}{ll}
v_j-p_j & \textrm{if $I_j$ wins in the auction}\\
0 & \textrm{otherwise}
\end{array} \right.
\tag{1}
\end{equation*}

An auction is deemed as \emph{truthful} if revealing the true valuation is the dominant strategy for each job $b_j$, regardless of other jobs' bids. More formally, an auction is truthful or strategy-proof, if for any job $I_j$ , and for all $f_j \ne v_j$ the following inequality always holds:

\begin{small}
\begin{equation*}
u_j (v_j , \mathbf{f}_{-j}) \geq u_j(f_j , \mathbf{f}_{-j})
\tag{2}
\end{equation*}
\end{small}
In other words, that an auction is truthful implies that no player can improve its own profit (utility) by bidding untruthfully. In our problem, truthfulness requires that:
\begin{compactenum}
\item The secondary users report their true valuations for the usage of spectrum channels (\emph{called \textbf{value-SP}}).

\item The secondary users report their true required time durations (\emph{called \textbf{time-SP}}).
\end{compactenum}

It has been proved in \cite{nisan2007algorithmic}, a bid monotonic auction is truthful and individually rational if it always charges \emph{critical values} from secondary users. The monotone allocation implies there is a \emph{critical value} for each job such that if $I_j$ bids higher than \emph{critical value} then it wins and if $I_j$ bids lower than \emph{critical value} then $I_j$ loses.

In this paper, we target at designing a heterogeneous spectrum auction which guarantees to achieve maximization of social efficiency under truthfulness constraint. Social efficiency is introduced to evaluate the performance of the proposed mechanism. The social efficiency for an auction mechanism $\mathcal{M}$ is defined as total true valuations of all winners, $i.e.$ $\text{EFF}(\mathcal{M})= \sum_{I_j \in \mathcal{I}} v_j y_j$, where $v_j$ denotes the true valuation of job $I_j$ and $y_j$ indicates whether the required channel is allocated to $I_j$. Hence, we concerned the following optimization problem:


\setcounter{equation}{2}
\parbox{5cm}{\begin{eqnarray*} \max \sum_{I_j \in \mathcal{I}} && v_j y_j; \nonumber\\
s.t. & & \emph{Economic Robust Constraints} \end{eqnarray*}} \hfill
\parbox{1cm}{\begin{eqnarray} \end{eqnarray}}

\section{VCG-based Optimal Auction Mechanism Design}\label{sec3}

In this section, we present an optimal auction which maximizes the expected social efficiency while enforcing truthfulness.

\subsection{VCG-based Optimal Auction Model}

Recalling our assumptions in the spectrum auction model section, auctions in different local markets have no mutual effects with each other, thus we just focus on the auction in a single local market $M_{k}$ in this section. Let $\mathcal{I}=\{I_1 ,...,I_N \}$ denote the job set in local market $M_k$, and $\mathcal{S}=\{s_1 ,...,s_M \}$ be the available spectrum (channel) set in $M_k$. Our aim is to achieve the maximization of the social efficiency through an optimal matching between sets $\mathcal{I}$ and $\mathcal{S}$.

Assume $A_i =\{x_{1,i} ,...,x_{q,i} \}$ includes all the available time slots in $s_i $. Each $I_j \in \mathcal{I}$ can only be allocated in the time slot of $s_i $ between $a_j $ and $d_j $. In order to simplify the matching model between $I_j $ and $s_i $, we will make a further segmentation to $A_i $ based on the arrival time and deadline of all the jobs in $\mathcal{I}$.
For each $I_j \in \mathcal{I}$, its arrival time/deadline divides one of the time slot in $s_i$ into 2 time slots. As shown in Fig. \ref{fig:segmentation}, the time axis of $s_i $ is divided into many disjoint time slots after our segmentation. Let $x_{l,i} $ be the \textit{l-th} time slot in $s_i $ and $\Delta _{l,i}$ be the length of $x_{l,i} $. We define $\Delta _{l,i} =0$ when time slot $x_{l,i} $ is occupied by the primary user. Assume that the time slot beginning at $a_j $ is the $_{ }n_s^{i,j} $-th time slot in $s_i $ and the time slot ending at $d_j $ is the $n_e^{i,j} $-th time slot in $s_i $. Formally, $y_{i,j} \in \{0,1\}$ is a binary variable indicating whether $I_j$ is allocated in $s_i$. We can formulate the spectrum assignment problem for jobs into an IP (\emph{Integral Programming}):

{\setlength{\abovedisplayskip}{1pt}
\setlength{\belowdisplayskip}{1pt}
\begin{small}
\begin{equation*}
\setlength{\abovedisplayskip}{1pt}
\setlength{\belowdisplayskip}{1pt}
\max  O(v)=\sum\limits_{I_j \in \mathcal{I}} {\sum\limits_{s_i \in \mathcal{S}} {v_j y_{i,j} } }
\tag{4}
\end{equation*}
\end{small}
subject to
\begin{small}
\begin{equation*}
\begin{cases}
y_{i,j} \in \{0,1\}, \forall i, \forall j  \\
\sum\limits_{s_i\in \mathcal{S}} {y_{i,j} } \le 1, \forall j \\
v_j \ge \eta_s t_j, \forall j \\
x_{l,i}^j \ge 0, \forall l, \forall i, \forall j\\
\sum_{l=n_s^{i,j} }^{n_e^{i,j} } {x_{l,i}^j } \ge t_j y_{i,j} ,\ \mbox{}\forall i, \forall j\\
\sum\limits_{I_j\in \mathcal{I}} {x_{l,i}^j \le \Delta _{l,i} }, \forall i, \forall l
\end{cases}
\end{equation*}
\end{small}}

where $x_{l,i}^j $ is the time $x_{l,i} $ allocated to $I_j $, $O(v)$ denotes the objective function of the IP.

\begin{figure}[!t]
\centering
\includegraphics[width=8cm]{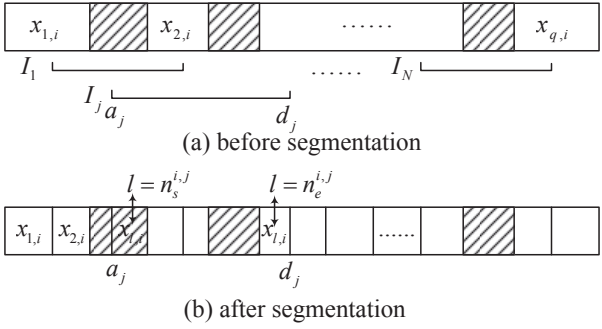}
\caption{An instance of the spectrum $s_i$'s time axis segmentation. Let shadow slots denote the time intervals occupied by primary user, and the blank slots indicate the intervals which can be allocated to secondary users.}
\label{fig:segmentation}
\end{figure}

\subsection{VCG-based Optimal Auction Design}

We first introduce a VCG-based optimal auction, which solves the objective function of \textbf{IP} (4) optimally. The winner determination is to maximize the social efficiency and the payment for each job is the opportunity cost that its presence introduces to all the other jobs. The detailed auction process is given in Algorithm \ref{alg:Algorithm 1}.\\

\vspace{-0.4cm}
\begin{algorithm}
\caption{VCG-based Optimal Auction Mechanism}\label{alg:Algorithm 1}
\begin{algorithmic}[1]

\STATE Let $\mathbf{X}^* = (x_1, x_2, ...,x_N)$ be the winner determining vector, where  $x_j =1$ means job  $I_j$ wins in the auction, while $ x_j=0 $ indicates no allocation for $I_j$;
\STATE Let $\mathbf{P}^* = (p_1 , p_2,..., p_N) $, where $p_j$  is the money that $I_j$  intends to pay the primary user;
\STATE Use VCG-based mechanism to get $\mathbf{X}^* $  and $\mathbf{P}^*$;
\\  \begin{enumerate}
	\item VCG-based mechanism includes an allocation to maximize the social efficiency.
\[
\mathop {\max }\limits_{\mbox{X}^\ast } \sum\limits_{I_j \in \mathcal{I}} {\sum\limits_{s_i \in \mathcal{S}} {v_j y_{i,j} } }
\]
\begin{center}
s.t. Allocation Constraints.
\end{center}
    \item Payment charges each job described as following:
\[
p_{j}^{\ast} =\mathop{\max }\limits_{\mbox{X}^\ast_{-j}} \sum\limits_{k\ne j}{\sum\limits_{s_i \in \mathcal{S}} {v_k y_{i,k}}} - \mathop{\max}\limits_{\mbox{X}^\ast} \sum\limits_{k\ne j} {\sum\limits_{s_i \in \mathcal{S}} {v_k y_{i,k}}}
\]
\[
p_j = \mathop{\max }(p_{j}^{\ast}, \eta_s t_j)
\]
	\end{enumerate}
\STATE The final allocation $\mathbf{X} $  is set to $\mathbf{X}^* $  and the payment $p_j$ is set to $p_{j}^{*} $  for each winner and $p_j =0$  for each loser.
\end{algorithmic}
\end{algorithm}

Solving the above IP optimally is an NP-hard problem, and its computational complexity is prohibitive for large scale spectrum market. Therefore, the VCG-based optimal auction mechanism is only suitable for a relative small-scale spectrum trading market. In Section \ref{sec4}, we will further design an auction that is truthful, but computes only an approximately optimal solution to maximize the social efficiency.

\subsection{Theoretical Analysis}

As mentioned before, economic robust constraint should be satisfied in auction design. We now analyze the properties of the optimal auction in terms of \emph{value-SP} and \emph{time-SP}. Recall that a bid monotonic auction is value-SP if it always
charges critical values from secondary users. In our VCG-based optimal auction, the payment for each job $I_j$ is decided by a fixed $\mbox{X}^\ast_{-j}$ or its requiring time. So we can say that $p_j$ is independent of $v_j$, which means it is a critical value.
Therefore, the requirement of critical value is satisfied immediately.

\begin{lemma}
\label{lemma:critical-value}
If $I_j$ wins, its payment $p_j$ is a critical value.
\end{lemma}

\begin{lemma}
\label{lemma:monotone}
The VCG-based optimal auction is bid monotone. That is, for each $I_j \in \mathcal{I}$, if $I_j $ wins by bidding price $f_j $, then it also wins by bidding any price ${f}'_j \geq f_j $.
\end{lemma}
\begin{proof}
Suppose $O^\ast (v)$ is the optimal solution of objective function (4) and $\mbox{X}^\ast$ is the winner determining vector corresponding to $O^\ast (v)$ when $I_j$ bids $f_j $. If $I_j $ wins in the auction, bidding higher can only increase the value of $O^\ast (v)$. Hence, $O^\ast (v)$ is also the optimal solution for objective function (4), $I_j$ always wins if it bids ${f}'_j \ge f_j $, the Lemma holds.
\end{proof}

According to Lemma \ref{lemma:critical-value} and Lemma \ref{lemma:monotone}, it follows that:
\begin{theorem}
\label{theo:vcg}
The VCG-based optimal auction is value-SP for secondary users.
\end{theorem}

Since the time-domain is introduced into spectrum auction mechanism design, time-SP issue should also be considered at the same time. Now, we will show that the proposed VCG-based optimal auction is time-SP for each secondary user.
\begin{theorem}
\label{theo:vcg-time}
The VCG-based optimal auction is time-SP for secondary users.
\end{theorem}
\begin{proof}
We assume each job $I_j \in \mathcal{I}$ could only claim a longer job length $t_j $ than its actual requirement. Since the bidding price is the same, if $I_j $ wins by bidding ${t}'_j \geq t_j $, it always wins by bidding truthfully. The payment of $I_j $ is time-independent or increases with $t_j$, so the utility of $I_j $ will not increase after it lies. However, $I_j $ may lose by bidding ${t}'_j \ge t_j $, while wins by bidding truthfully. In this case, the utility of $I_j $ will decrease after it lies.
This finished the proof.
\end{proof}

Based on Theorem \ref{theo:vcg} and Theorem \ref{theo:vcg-time},
we have proved that the designed VCG-based optimal auction is truthful for secondary users.

Since that our allocation model is discrete, in what follows, we would like to investigate whether our solution can be simulated by any continuous model. Let the optimal performances under discrete and continuous allocation models be denoted as $P_1$ and $P_2$ respectively. In fact, we have the following conclusion:

\begin{theorem}
\label{theo:performance}
$\frac{P_1}{P_2} \rightarrow \infty$.
\end{theorem}
\begin{proof}
Assume there exists only one spectrum (channel) $s_i$ and a single job $I_j$ in the scenario. Consider the case as shown in Fig. \ref{fig:comparison}, if $\Delta_{l,i}, \Delta_{l+1,i}< t_j$, $\Delta_{l,i}+ \Delta_{l+1,i}\geq t_j$, $I_j$ cannot be allocated in $s_i$ with continuous time slots. Therefore, $P_2 =0$ holds. However, $I_j$ can be accepted in the discrete allocation model, thus we get $P_1 = v_j$. So $\frac{P_1}{P_2} \rightarrow \infty$, and the theorem holds.
\end{proof}

\begin{figure}[!t]
\centering
\includegraphics[width=2in]{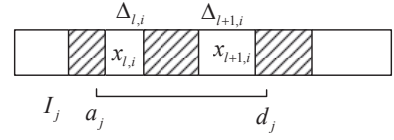}
\caption{An instance of the spectrum market in which exists only one spectrum $s_i$ and a single job $I_j$.}
\label{fig:comparison}
\end{figure}

\section{Suboptimal Auction Mechanism Design}\label{sec4}

In practice, achieving the social efficiency problem optimally is infeasible in large scale spectrum market. Therefore, we present a more computationally efficient Per-Value Greedy (PVG) auction mechanism in this section. The PVG auction first allocates the jobs for approximately maximizing social efficiency with a greedy allocation mechanism, and then charges the critical value of jobs to ensure truthfulness.

\subsection{Spectrum Allocation Mechanism}

We now outline the greedy allocation mechanism. Recall that $v_j$ is the weight (bid) of job $I_j $ and $t_j$ is the job length of $I_j$. The per-unit weight (bid) of job $I_j $ can be calculated through $\eta _j=\frac{v_j }{t_j}$. All feasible jobs in $\mathcal{I}$ are sorted in descending order according to $\eta _j $. The Algorithm \ref{alg:2} maintains a set $\mathcal{A}$ of currently accepted jobs. There are three possible cases that job $I_j $ can be accepted by Algorithm \ref{alg:2}.

Case 1: When a new job $I_j $ is considered according to the sorted order, we scan all the available spectrums one by one, $I_j $ is immediately accepted if it can be allocated in one of these spectrums without time overlapping with any other jobs in $\mathcal{A}$.

Case 2: If $I_j $ overlaps with some jobs in $\mathcal{A}$, it can also be accepted when its weight is larger than $\beta (\beta \geq 1)$ times the sum of weights of all overlapping jobs. In this case, we add $I_j $ in $\mathcal{A}$ and delete all the overlapping jobs in $\mathcal{A}$, and say that $I_j $ ``\textbf{preempts}'' these overlapping jobs.

Case 3: After some other jobs accepted in Case 2, job $I_j$ which has been rejected or preempted before can be reconsidered for acceptance if it can be allocated without overlapping with any other jobs.

Define $J_k $ as the job with the \emph{k}-th smallest $\eta _k$ which is allocated in the time slots of $s_i$ between $a_j$ and $d_j$. We say that jobs $J_1 ,...,J_h $ overlap with $I_j$, if $I_j$ can be allocated without time overlap by deleting jobs $J_1 ,...,J_h $ in $\mathcal{A}$, but cannot be allocated by deleting jobs $J_1 ,...,J_{h-1} $ in $\mathcal{A}$.

If $I_j $'s weight is no larger than $\beta (\beta \ge 1)$ times the sum of weights of all overlapping jobs $J_1 ,...,J_h $, we say that jobs $J_1
,...,J_h $ directly ``\textbf{reject}'' $I_j $.

A job $I_j $ ``\textbf{caused}'' the rejection or preemption of another job $J$, if either job $I_j $ directly rejects or preempts job
$J$, or preempts $J$ indirectly. For example, if job $I$ is preempted by job $J$ and job $J$ is preempted by $I_j$, we say that $I_j $ preempts
$I$ indirectly.

If $I_j$ is accepted, we allocate time slots for job $I_j$ starting from its arrival time and searching for a series of available time slots in a backward manner. The approximation factor of our allocation mechanism is stated in the following.\\

\begin{algorithm}
\caption{Spectrum Allocation Algorithm}\label{alg:2}
\begin{algorithmic}[1]

\renewcommand{\algorithmicrequire}{\textbf{Input:}}
\renewcommand{\algorithmicensure}{\textbf{Output:}}

\REQUIRE ~~\\

    $\mathcal{I}=\{I_1,...,I_N\} $ // $\mathcal{I}$: the set of all the jobs in $M_k$ sorted in descending order according to $\eta_j$;\\
    $\mathcal{S}=\{s_1,...,s_M\} $ // $\mathcal{S}$: the set of available spectrum in $M_k$;\\

\ENSURE ~~\\

    The set of accepted jobs in  $\mathcal{A}$;

\STATE $\mathcal{A}=\emptyset$;
\FOR {$j = 1$ to $N$}
\IF {$v_j \ge \eta_s t_j$ }
\FOR {$i=1$ to $M$}
\IF {$I_j$ can be allocated in $s_i$ that not overlap with other jobs}
\STATE $\mathcal{A}:=\mathcal{A}\cup\{I_j\} $;
\STATE \textbf{accept }$I_j$ and allocate it in $s_i$;
\STATE \text{Break}
\ENDIF
\ENDFOR
\IF {$I_j \notin \mathcal{A}$}
\FOR{$i$= \text{1 to $M$}}
\IF {the total weight of jobs $\{J_1,...,J_n\}$ that overlap with $I_j$ is smaller than $w/ \beta $}
\STATE $\mathcal{A}:=\mathcal{A}\cup\{I_j\}/\{J_1,...,J_n\} $;
\STATE \textbf{preempt} $\{J_1,...,J_n\}$ and allocate $I_j$ in $s_i$;
\FOR {$k=1$ to $j$}
\IF {$I_k \notin \mathcal{A}$, $v_k \ge \eta_s t_k $  and can be allocated in $s_i$ that not overlap with other jobs}
\STATE $\mathcal{A}:=\mathcal{A} \cup \{I_k\}$;
\STATE \textbf{accept} $I_k$ and allocate it in $s_i$;
\ENDIF
\ENDFOR
\STATE \text{Break}
\ENDIF
\ENDFOR
\ENDIF
\IF {$I_j \notin \mathcal{A}$}
\STATE \textbf{reject } $I_j$;
\ENDIF
\ENDIF
\ENDFOR
\RETURN $\mathcal{A}$;
\end{algorithmic}
\end{algorithm}

\begin{theorem}
\label{theo:appfactor}
The approximation factor of the PVG is $6+4 \sqrt{2}$.
\end{theorem}
\begin{proof}
Let $\mathcal{O} $ be the set of jobs chosen by the optimal mechanism $\mathcal{OPT}$. Let the set of jobs accepted by Algorithm \ref{alg:2} be denoted by $\mathcal{A}$. For each job $I \in \mathcal{A} $, we define a set $R(I)$ of all the jobs in $\mathcal{O}$ that ``accounted for" by $I$. $R(I)$ consists of $I$ if $ I \in \mathcal{O}$, and all the jobs in $\mathcal{O}$ which are rejected or preempted by $I$. More formally:
\begin{enumerate}
\item Assume $I$  is accepted by case 1 or 3, then $R(I)=\{I\}$ in the case of $I \in \mathcal{O}$, and $R(I)=\emptyset$ otherwise.
\item Assume  $I$ is accepted by case 2, then  $R(I)$ is initialized to contain all those jobs from $\mathcal{O}$  that were preempted (directly or indirectly) by $I$. In addition, $R(I)$  contains $I$  in the case of $I \in \mathcal{O}$.
\item Assume $J \in \mathcal{O}$  is rejected by line 20 in Algorithm \ref{alg:2} and let $I_1,...,I_h$ be the jobs in $\mathcal{A}$  that overlapped with $J$  at the time of rejection. We can only allocate each job in the same spectrum in our model. Hence,  $I_1,...,I_h$ and $J$  are allocated in the same spectrum. Let $v$ denote the weight of $J$  and let $v_j$ denote the weight of  $I_j$ for $1 \leq j \leq h $. We view $J$  as $h$  imaginary jobs $J_1 ,..., J_h $, where the weight of $J_i$  is $\frac{v_j v}{\sum_{j=1}^{h} v_j } $  for $1 \leq j \leq h $. $R(I_j):= R(I_j) \cup \{J_j\} $. Note that the weight of $J_j$ is no larger than $\beta$ times the weight of $I_j$ according to the rejection rule.
\end{enumerate}

For each job $J \in \mathcal{O}$, if $J \in \mathcal{A}$, it had been included in $R(J)$  through our acceptance rule; otherwise, it must be preempted or rejected by some $I \in \mathcal{A}$. In this case,  $J$ belongs exactly to the sets $R(I)$  that preempted or rejected by $I$. Thus, the union of all sets  $R(I)$ for $I \in \mathcal{A}$ covers $\mathcal{O}$.

We now fix a job  $I \in \mathcal{A}$. Let  $v$ be the weight of $I$  and let $V$  be the sum of weights of all jobs in $R(I)$. Then, we can get that $V=v'+v''+v$  if $I \in \mathcal{O}$, where $v'$  denotes the weights of all jobs preempted by $I$,  $v''$ is the weights of all jobs or portion of them rejected by $I$; otherwise, $V=v'+v''$. Therefore, we can conclude that $V \leq v'+v''+v$. Define $\rho = V/v$. Our goal is to give the upper bound of $\rho$.

We first consider the jobs that have been rejected by $I$. According to line 20, if $J \in \mathcal{O}$ overlaps with jobs $I_1 ,...,I_h$, we split $J$  into $h$  imaginary jobs $J_1 ,...,J_h$, and let each overlapping job $I_j$ account for an imaginary job $J_j$. Therefore, we can assume that each $J_j$ is only to be rejected by $I_j$. On the other hand, if we remove $I$  from $\mathcal{A}$, all of the jobs or imaginary jobs rejected by $I$ can be accepted by $\mathcal{A}$.

Let $J_1,...,J_q$ be the jobs in $\mathcal{O}$ which were rejected by $I$. If there exists a job $J_k$ ($1 \leq k \leq q$) which can partition $J_1 ,...,J_q$  into two disjoint sections in time-axis, we will define the arrival point of $J_k$ a critical point; otherwise, we will choose the arrival time of job $I$ as the critical point. The whole time-axis will be classified into \emph{LOS} (Left of Separatrix) and \emph{ROS} (Right of Separatrix) by using the critical point as the separatrix. We define the job whose arrival time later than critical point belongs to \emph{ROS}; otherwise, the job belongs to \emph{LOS}.

Assume that the job length of $I$ is $t$, the allocated time in \emph{LOS} and \emph{ROS} last $t'$ and $t''$ respectively. We can easily get that:

\vspace{-9pt}
\begin{small}
\begin{equation*}
t= t'+t''
\tag{5}
\end{equation*}
\end{small}

Assume that there exists two jobs $J_1$  and  $J_2$ located at \emph{ROS} in Fig. \ref{fig:LOS}, and  $a_1$ is the critical point. Since we allocate time slots for job  $J$ starting from its arrival time and searching for continuous available time slots in a backward manner, $J_2$ rejected by $I$  indicates that all the available time from $a_2$  to $d_2$  is less than its job length when $d_1 \leq d_2$ . However,  $J_1$ and  $J_2$ can be accepted while removing $I$  from $\mathcal{A}$ , hence the job length of $J_1$  is equal to the time overlapped with the jobs allocated between $a_1$  and $d$. Assume that the overlap time $J_k$ with job $I$ is $t'_k$.  If all the jobs account for the weight of time overlap with $J_1$ , we can cover the total weight of $J_1$. Therefore, the weight of $I$ should account for $J_1$ is equal to $\eta_1 t'_1 $. It's obvious that $ t'_1 < t''$, thus we can calculate the total weight of all the jobs $I$ should account for in \emph{ROS}:

\begin{small}
\begin{equation*}
V_{ROS} \leq v_2 + \eta_1 t'_1 \leq v_2 + \eta_1 t''
\tag{6}
\end{equation*}
\end{small}

When $d_1 > d_2$ , we demonstrate that $V_{ROS} \leq v_1 + \eta_2 t'_2 \leq v_1 + \eta_2 t''$ similarly.

Assuming there are more than two jobs in $\mathcal{A}$  located in ROS which were rejected by $I$ , we can easily conclude that $V_{ROS} \leq v_i + \sum_{k \ne i} \eta_k t'_k$, and $\sum_{k \ne i}t'_k \leq t'' $.

According to the rejection rule, we can get:

\begin{small}
\begin{equation*}
v_i \leq \beta v
\tag{7}
\end{equation*}
\end{small}
Since the $\eta_k$ from any of the rejected job $J_k$ is less than the value of $\eta$ from $I$, the following holds:

\begin{small}
\begin{equation*}
\sum_{k \ne i} \eta_k t'_k \leq \sum_{k \ne i} \eta t'_k \leq \eta t''
\tag{8}
\end{equation*}
\end{small}

Based on (7) and (8), the sum of weights $V_{\text{ROS}}$  of all the rejected jobs in ROS satisfies:

\begin{small}
\begin{equation*}
V_{ROS}\leq \beta v + t''\eta
\tag{9}
\end{equation*}
\end{small}

By the same method, we can easily find sum of weights  $V_{\text{LOS}}$ of all the rejected jobs in LOS satisfies:
\begin{small}
\begin{equation*}
V_{LOS}\leq \beta v + t'\eta
\tag{10}
\end{equation*}
\end{small}

Combining the above two conditions, sum of weights of all jobs can be calculated:
\begin{small}
\begin{equation*}
v'' = V_{ROS}+V_{LOS} \leq 2\beta v + (t' +t'')\eta=2\beta v +v
\tag{11}
\end{equation*}
\end{small}
\begin{figure}[!t]
\centering
\includegraphics[width=7cm]{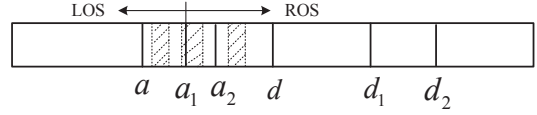}
\caption{The whole time-axis is classified into \emph{LOS} (Left of Separatrix) and \emph{ROS} (Right of Separatrix) by using the critical point as the separatrix. Arrival time $a_1$ of the job $J_1$ is the critical point in this instance. Shadow parts indicate the time slots preempted by $I$.}
\label{fig:LOS}
\end{figure}

We now assume inductively that the $\rho$  bound is valid for jobs with a larger per-unit weight than that of $I$. Since the overall weight of the jobs that directly preempted by $I$ is at most $v/\beta$, we can get $v/\beta *\rho \geq v' $. Recall that $V \leq v' +v'' +v $  and  $v'' \leq 2\beta v + v $ hold. We can obtain that $V \leq v\rho / \beta + 2\beta v + v +v $. This implies that $V/v = \rho \leq \rho / \beta + 2\beta +2 $. The inequality can be depicted as  $\rho \leq \frac{2(\beta +1)}{1- 1/\beta } $ equivalently. This inequality takes its minimal value when  $\beta = 1+\sqrt{2} $, which implies that $\rho \leq 6+ 4\sqrt{2} $. Finally, since the $\rho$  bound holds for all the jobs in $\mathcal{A}$  and the union of all $R(I)$  sets covers all the jobs taken by $\mathcal{OPT}$, we can conclude that the $EFF(\mathcal{OPT})$ is at most $\rho$  times the social efficiency $EFF(\mathcal{A})$. Therefore, the approximation factor is $6+4 \sqrt{2} $.
\end{proof}

\subsection{Payment Calculation}

After getting the set of accepted jobs by Algorithm \ref{alg:2}, we will calculate the payment for each winner. An auction is value-SP if and only if it is bid monotone and always charges the winners its critical value. Therefore, we use the \emph{binary search} to find the critical value for each job in $\mathcal{A}$.

Let $p_j$ denote the payment of job  $I_j$, and $p'_j$ denote the critical value of $I_j\in\mathcal{A}$ calculated through \emph{binary search}. Since the payment for winner $I_j$ should be no less than $\eta_s t_j$,  $p'_j$ satisfies:

\begin{small}
\begin{equation*}
\eta_s t_j \leq p'_j \leq v_j
\tag{12}
\end{equation*}
\end{small}

We charge $p_j ={p}'_j $ for each winner and $p_j =0$ for each loser.
According to the payment rule, we can easily get that:

\begin{lemma}
\label{lemma:pvgcriticalvalue}
If $I_j$ wins, its payment $p_j $ is a critical value in the PVG auction.
\end{lemma}
\subsection{\textbf{Theoretical Analysis of the PVG Auction}}

Similar to the analysis of the VCG-based optimal auction, we first prove the most important economic property strategyproofness of the PVG auction which requires both \emph{value-SP} and \emph{time-SP}.

To prove the \emph{value-SP} of the PVG auction, we should first prove that the allocation resulting from the PVG auction is bid monotone.

\begin{lemma}
\label{lemma:pvgmonotone}
The allocation resulting from the PVG auction is bid monotone.
\end{lemma}
\begin{proof}
Supposing job $I$ wins by bidding $f$ in the PVG auction, we increase $I$'s bid from $f$ to ${f}'$ with ${f}'>f$. Assume $j$ is the rank of $I$ by
bidding $f$ in $\mathcal{I}$, and $i$ is the new rank of $I$ by bidding ${f}'$. Then we can get that $\eta =\frac{f}{t}\le \frac{{f}'}{t}={\eta }'$,
$i\leq j$.

Suppose $I$ loses in the auction by bidding ${f}'$, there are two possible cases:

Case 1: The job $I$ is accepted through line 6, but preempted by job $J$, This means $I$ overlaps with $J$. According to our preemption rule, $J$ also can be accepted when $I$ bids truthfully. There is no enough time available for $I$ after $J$ is accepted. Thus, $I$ cannot be accepted by bidding $f$ either. If $J$ is also be preempted by another job, and if $I$ is still not be accepted by bidding $f'$ according to line 17, then it will not be accepted by bidding $f$ either.

Case 2: $I$ has never been accepted. This
means $I$ overlaps with one or more jobs which have been accepted in $\mathcal{A}$.
In this case, $I$ cannot be accepted by bidding $f$ either.

According to the above analysis, if $I$ wins by bidding $f$ in the PVG
auction, it always wins by bidding ${f}'>f$. Therefore, the allocation
resulting from the PVG auction is bid monotone.

The lemma holds.
\end{proof}

We will give the truthful demonstration of the proposed PVG auction mechanism through Theorem \ref{theo:pvgvaluesp} and Theorem \ref{theo:pvg-time}. With Lemma \ref{lemma:pvgcriticalvalue} and Lemma \ref{lemma:pvgmonotone}, we obtain:

\begin{theorem}
\label{theo:pvgvaluesp}
The PVG auction is value-SP.
\end{theorem}

Then we show that the PVG auction scheme is time-SP for secondary users.
\begin{theorem}
\label{theo:pvg-time}
The PVG auction is time-SP.
\end{theorem}
\begin{proof}
We also assume the job $I_j \in I$ can only claim a longer job length $t'_j$ than its actual requirement in the \emph{PVG} auction. If $I_j$ wins by bidding a fake job length $t'_j > t_j$, we can get that:

\begin{small}
\begin{equation*}
\eta_j = \frac{f_j}{t_j} \leq \frac{f_j}{t'_j}=\eta'_j
\tag{13}
\end{equation*}
\end{small}

Similar to the proof of Lemma \ref{lemma:pvgmonotone}, we can easily obtain that $I_j$ also wins by bidding $t_j$. Since the payment of $I_j$ is time-independent or increased with $t_j$, since the utility of $I_j$ will not be increased when it lies. However, $I_j$ may lose by bidding $t'_j > t_j$, while winning by bidding truthfully. In this case, the utility of $I_j$ will be decreased after it lies.

So Theorem holds.
\end{proof}

Based on Theorem \ref{theo:pvgvaluesp} and Theorem \ref{theo:pvg-time}, we have proved that the PVG auction is also truthful for secondary user.

At last, we discuss the time complexity of the PVG auction. Assume $\xi$ denotes the minimum bid size, and $V_{max} = \mathop{\max}\limits_{I_j \in \mathcal{A}}(v_j - \eta_s t_j)$. Therefore, we have:

\begin{theorem}
\label{theo:complexity}
The time complexity of the PVG auction is $O(MN^2 logP)$, where $M$ is the number of channels, $N$ is the number of jobs, and $P=V_{max}/\xi $.
\end{theorem}

\section{Simulation Results}\label{sec5}
In this section, we conduct extensive simulations to examine the performance of the proposed auctions.

\subsection{Simulation Setup}

In order to make the experimental results more convincing and close to reality, we adopt the data set based on analysis of measurement data, which is collected in Guangdong Province, China. We choose the frequency band of \emph{Broadcasting TV1}(48.5 - 92MHz) for comparison from many frequency band of services, and intercept continuous 5 days' record from the whole measurement data. Total bandwidth of \emph{TV1} is split into plenty of channels in accordance with the width of 0.2MHz. For each channel, the data are divided into massive time slots, and we roughly set each time slot about 75 seconds. As a result, total number of time slots reaches to 5760 (5days/75s).

Fig. \ref{fig:tv1spec} shows a depiction of the channel vacancy located in frequency band of \emph{TV1}. We use black color to represent the occupied time slots and white color to denote the white space for each channel. The comparison figure makes some characteristics of spectrum usage easier to visualize, for instance, we can easily find that the usage time of primary user is basically the same in each day. Therefore, the vacancy time slots in all 5 days are selected as the idle slots for auction to ensure the usage of primary user at the same time.

In our simulations, we select 3 channels from the whole frequency band of \emph{TV1} as input, and the total time of each channel lasts 24 hours from 0:00 to 24:00.
We generate all the bid values from secondary users with a reservation price and the requirement of job length for each secondary user is uniformly distributed in the range of [0.5,2] hours. The request time interval with arrival time and deadline for each secondary user is uniformly distributed in the range of [2,4] hours. $\lambda$ shows the number of requests in our setting, Here, we generate two different scenarios.

\begin{itemize}
  \item \textbf{\emph{Set 1:}} All the requests are uniformly distributed in 24 hours without hot time.\\
  \item \textbf{\emph{Set 2:}} There exists hot time in this setting, which contains about $\delta$ requests of the whole day. In our simulation setting, we set $\delta=80\% $.
\end{itemize}

\begin{figure}[!t]
\centering
\includegraphics[width=8cm]{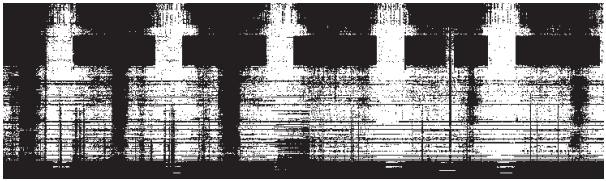}
\caption{Usage of spectrum for 5 days, an instance of frequency band of \emph{Broadcasting TV1}. }
\label{fig:tv1spec}
\end{figure}

\subsection{Performance of the Auction Mechanisms}

In this section, we study the performance of the PVG mechanism compared with the VCG-based optimal mechanism. We mainly focus on the performance of social efficiency and total revenue for primary user. For comparison, we plot results for 2 different request sets mentioned above, and analyze influences of the relationship between supply and demand from the results.

Fig. \ref{fig:socialefficiencyratio}(a) illustrates the social efficiency ratio of the PVG auction and the VCG-based optimal auction. We see that the PVG auction works as well as the VCG-based optimal auction when $\lambda$ is small. This is because there is enough available time for each secondary user, and most of them can be allocated without overlapping with others in both schemes. The competition among secondary users increases with $\lambda$, the VCG-based optimal auction outperforms the PVG auction gradually. The social efficiency ratio keeps approximately stable when $\lambda$ is large enough, where the supply is much less than demand. From Figs. \ref{fig:socialefficiencyratio}(a) and \ref{fig:spectrumutilizationratio}(b), we can see that the PVG auction performs best in the lightly loaded system and worst in highly loaded system of set 2. However, even in the worst case, the social efficiency ratio of the PVG auction is still above 70\% of the VCG-based optimal auction.


\begin{figure}[!t]
\centering
\begin{tabular}{cc}
\includegraphics[width=0.24\textwidth]{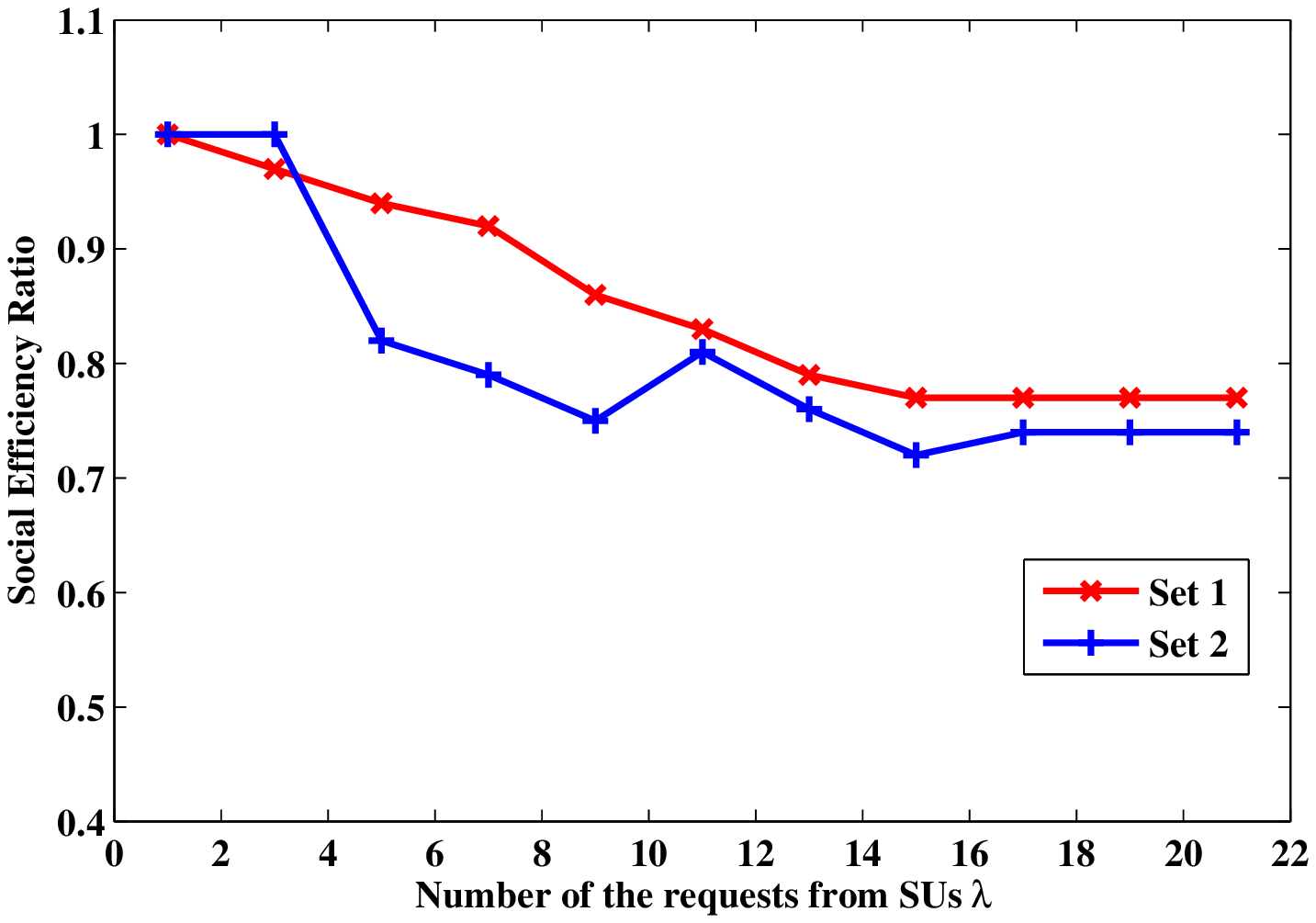} &
\includegraphics[width=0.24\textwidth]{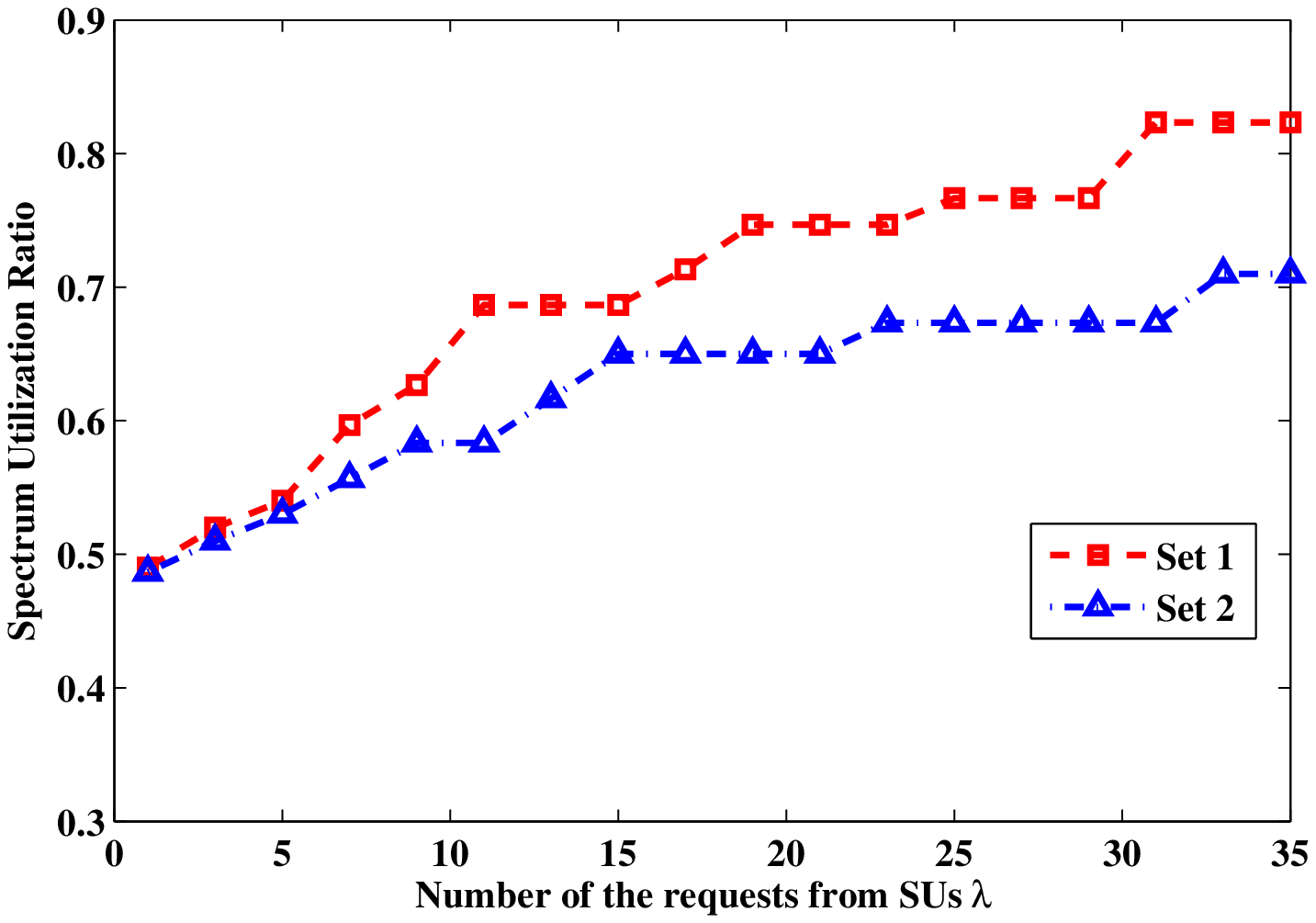}\\
(a) {\small Social Efficiency Ratio} &
(b) {\small Spectrum Utilization Ratio}
\end{tabular}
\caption{Social efficiency ratio and spectrum utilization ratio under Set 1 and 2, $\eta_s = 0$, $\beta=2$.}
\label{fig:socialefficiencyratio}
\label{fig:spectrumutilizationratio}
\end{figure}

\begin{figure}[!t]
\centering
\begin{tabular}{cc}
\includegraphics[width=0.24\textwidth]{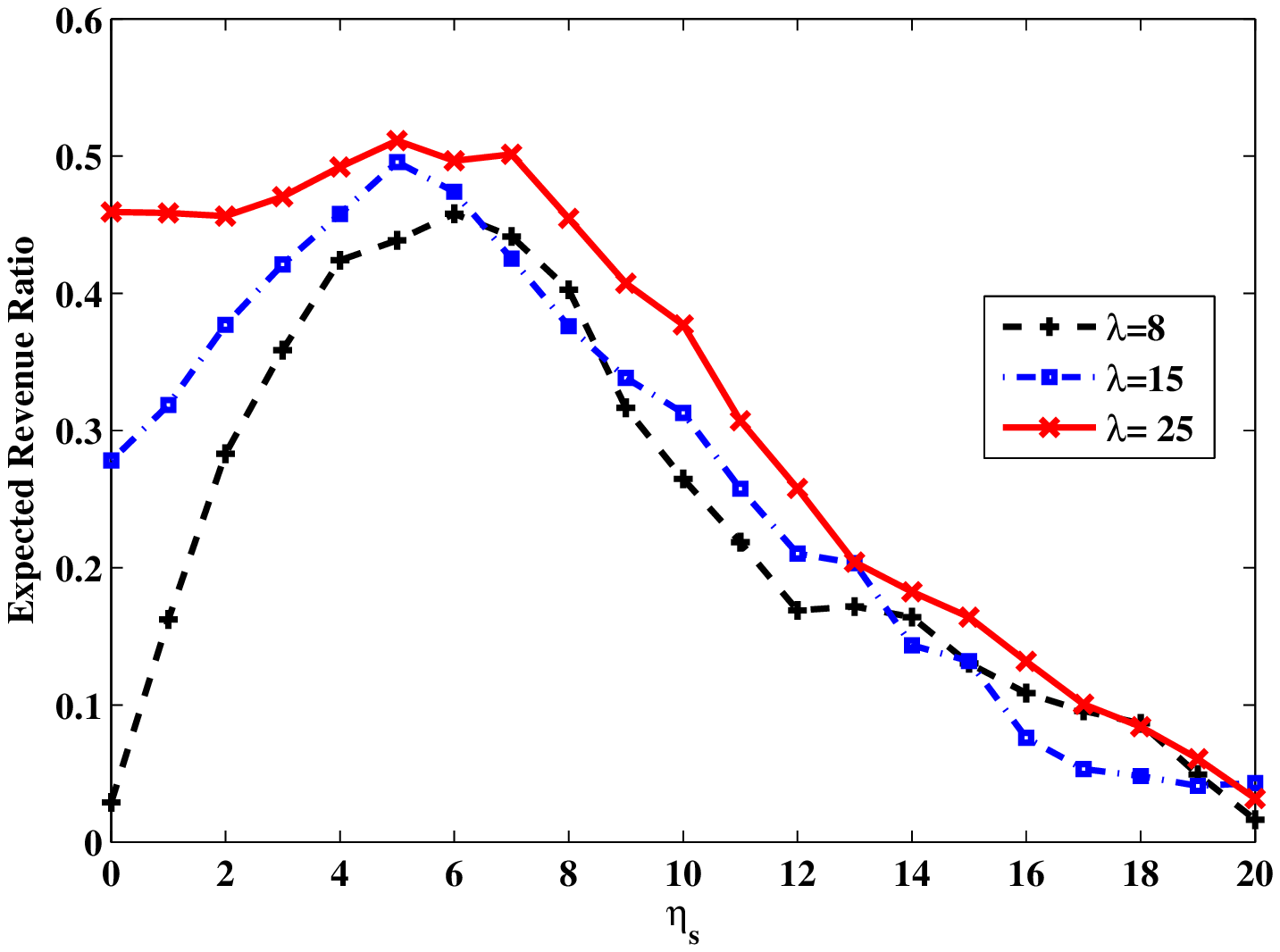} &
\includegraphics[width=0.24\textwidth]{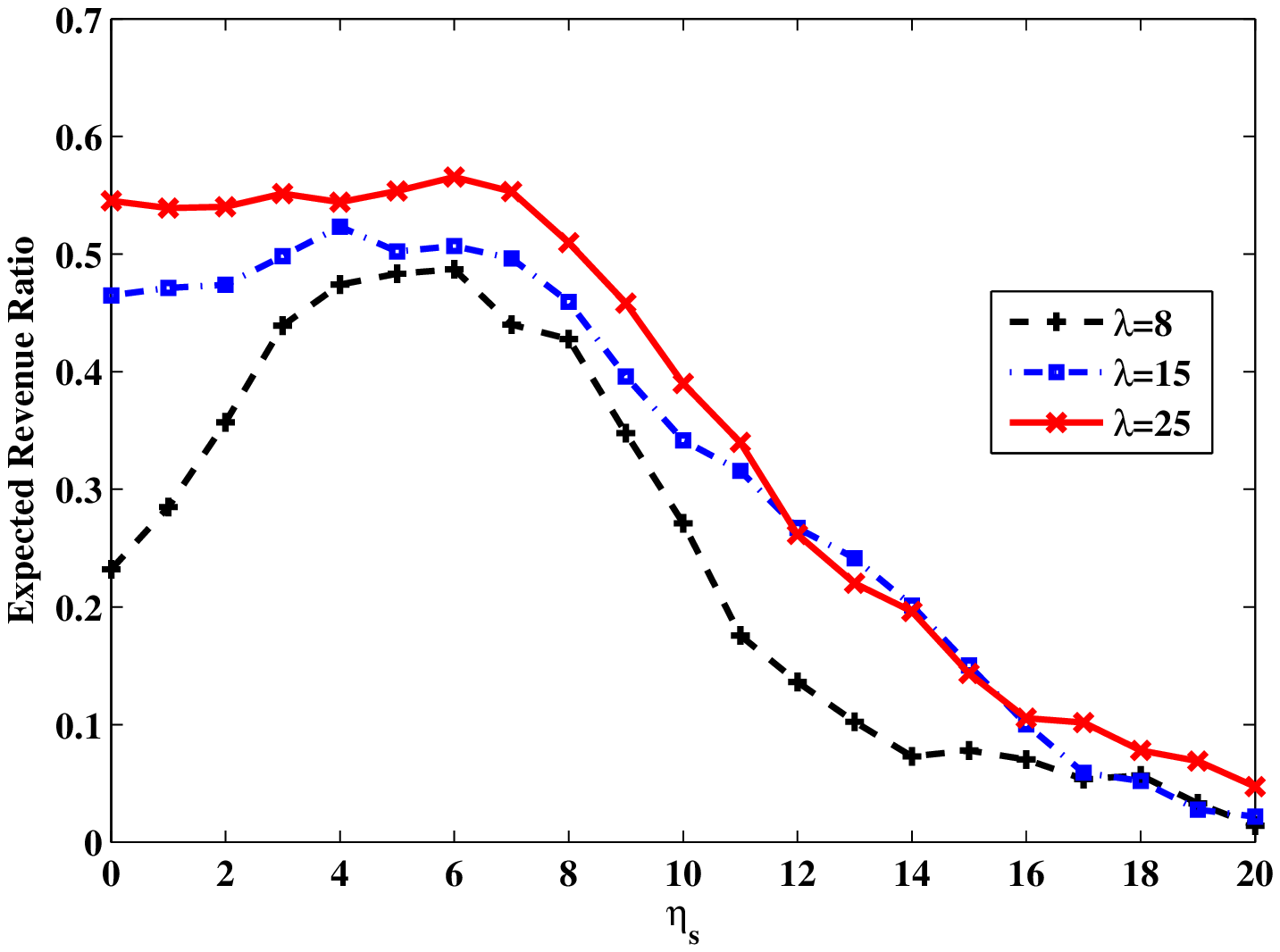} \\
(a) set 1 & (b) set 2
\end{tabular}
\caption{Expected revenue ratio under different data sets, $\beta=2$.}
\label{fig:revenuetv1}
\label{fig:revenuetv1hot}
\end{figure}


In Fig. \ref{fig:revenuetv1}(a), we plot the relationship between $\eta_s$ and the expected revenue ratio for primary user in Set 1, where the expected revenue ratio is the total payment of all the winning secondary users compared with the social efficiency when $\eta_s=0$. When the number of requests is small (8 or 15 in Fig. \ref{fig:revenuetv1}(a)), the competition is weak, some of the requests from secondary users can be allocated in channels without overlapping with others. Thus, the payment of these secondary users is equal to the product of $\eta_s$ and the request time. In this case, the total revenue for primary user increases along with $\eta_s$. However, a request cannot be allocated in channels if its per-unit bid is smaller than $\eta_s$. Thus, many secondary users may lose in the auction due to their bids are smaller than the product of $\eta_s$ and their request time. The revenue of primary user will decreases with the value of $\eta_s$ when $\eta_s$ is set too high. Most of the requests from the wining secondary users overlap with at least one request from a losing secondary user, when there are plenty of requests from secondary users in the spectrum market ($\lambda=25$). The revenue of primary user will decrease with the value of $\eta_s$ in this case. This is more obvious in Set 2. Due to the fierce competition among requests in hot time, the revenue of primary user doesn't increase by setting a higher $\eta_s$ even when $\lambda=15$.
We can make some reasonable hypothesis based on the analysis of experimental results from Fig. \ref{fig:revenuetv1}(a) and Fig. \ref{fig:revenuetv1hot}(b). The primary user can improve its revenue by setting a suitable $\eta_s$ under the condition that supply exceeds demand. The revenue of primary user can be maximized through the competition of secondary users in the condition that demand exceeds supply, the revenue will decrease with a large $\eta_s$.

\section{Literature Reviews}\label{sec6}
Auction serves as an efficient way to distribute scarce resources in a market and it has been extensively studied in the scope of spectrum allocation in recent years. 
Many works follow on the designs of wireless spectrum auctions in different scenarios. 
For instance, \cite{gandhi2007general} and \cite{ryan2006new} study the spectrum band auctions aiming to minimize the spectrum interference. 

Truthfulness is a critical factor to attract participation \cite{klemperer2002really}. 
Many well-known truthful auction mechanisms are designed to achieve economical robustness \cite{babaioff2001concurrent},\cite{mcafee1992dominant},\cite{vickrey1961counterspeculation}.
Unfortunately, none of the earlier spectrum auctions address the problem of truthfulness.
Truthfulness is first designed for spectrum auction in \cite{zhou2008ebay}, where spectrum reuse is considered. 
Similar model is adopted by the following works: \cite{zhou2009trust}, \cite{xu2010secondary}, \cite{wang2010toda}, \cite{wangdistrict}, \cite{huangsprite}, \cite{jia2009revenue}, \cite{al2011truthful}, \cite{gopinathanstrategyproof}, \cite{wu2011small}, \cite{wang2010toda}, \etc Specifically, \cite{jia2009revenue} and \cite{al2011truthful} focus on designing truthful mechanisms for maximizing revenue for the auctioneer; \cite{gopinathanstrategyproof} chooses the classic max-min fairness criterion in the study of the fairness issue in spectrum allocations to achieve global fairness; \cite{wu2011small} supports spectrum reservation prices in the auction model. TODA \cite{wang2010toda} first takes time domain into account, and proposes a truthful suboptimal auction with polynomial time complexity aiming to generate maximum revenue for the auctioneer. District mechanism \cite{wangdistrict} first takes the spectrum locality into account and gives an economically robust and computationally efficient method. Different from traditional periodic auction model, many works study the spectrum allocation in an online model \cite{wang2010toda},\cite{deek2011preempt},\cite{xu2011efficient}. However, most existing works concentrate on a truthful mechanism design without considering spectrums as non-identical items. The proposed optimal and sub-optimal spectrum auction mechanisms take the inherent spectrum heterogeneous characteristics into consideration in this paper.

Heterogeneous spectrum transaction issue has been studied in \cite{fengtahes},\cite{parzy2011non} and \cite{kash2011enabling}. In \cite{fengtahes}, Feng \emph{et al.} propose a truthful double auction method for heterogeneous spectrum allocation. \cite{parzy2011non} and \cite{kash2011enabling} solve the heterogeneous auction problems in different perspectives. However, they do not consider time domain issue in their works, thus making the spectrum allocation incomplete.

\section{Conclusion}\label{sec7}
In this paper, we  studied a general truthful secondary spectrum auction framework of heterogeneous spectrum allocation.
We  designed two auction mechanisms to maximize the social efficiency. One is optimal design for social efficiency by using the classic VCG mechanism, but it has high complexity. The proposed \emph{PVG} auction scheme has a constant approximation factor but is computationally much more efficient. These auctions provide primary users sufficient incentive to share their spectrum and make dynamic spectrum access more practical.
To the best of our knowledge, this is the first work that takes both the spectrum heterogeneity and a flexible time request from secondary users into consideration.

Several interesting questions are left for future research.
The first one is to study the case when the request of a secondary user may be served by several channels in a single local market.
The second one is when a secondary user's requests may cover multiple spectrum markets.
We need to investigate the impact of crossing dependence of different markets.
The third challenging question is to design truthful mechanisms when we have to make online decisions.

\section*{Acknowledgement}
We wish to thank the Prof. Shufang Li, Dr. Sixin Yin from BUPT for providing the spectrum usage data and Dr. Juncheng Jia for his insightful suggestions. We gratefully acknowledge the support from the National Grand Fundamental Research 973 Program of China (No.2011CB302905).

\ifCLASSOPTIONcaptionsoff
  \newpage
\fi

\bibliography{auction}

\end{document}